\documentclass[graybox]{svmult}
\usepackage{mathptmx}       
\usepackage{helvet}         
\usepackage{courier}        
\usepackage{type1cm}        
\usepackage{makeidx}         
\usepackage{graphicx}        
\usepackage{multicol}        
\usepackage[bottom]{footmisc}
\usepackage{multirow,amsmath,amssymb,url}
\usepackage{verbatim}

\makeindex             
\begin{document}

\title{$D$-optimal saturated designs: a simulation study}
\author{Roberto Fontana, Fabio Rapallo and Maria Piera Rogantin}
\institute{Roberto Fontana \at Department DISMA, Politecnico di Torino, Corso Duca degli Abruzzi 24, 10127 TORINO, Italy, \email{roberto.fontana@polito.it}
\and Fabio Rapallo \at Department DISIT, Universit\`a del Piemonte Orientale, Viale Teresa Michel 11, 15121 ALESSANDRIA, Italy, \email{fabio.rapallo@unipmn.it} \and Maria Piera Rogantin \at Department DIMA, Universit\`a di Genova, Via Dodecaneso 35, 16146 GENOVA, Italy, \email{rogantin@dima.unige.it}}
%
%
\maketitle

\abstract*{In this work we focus on saturated $D$-optimal designs.
Using a recent result in \cite{jspi2013}, we identify $D$-optimal
designs with the solutions of an optimization problem with linear
constraints where the objective function to be maximized is the
determinant of the information matrix. We study the possibility to
replace the determinant of the information matrix with simpler
objective functions that could give the same optimal solutions.
These new objective functions are based on the geometric structure
of the design. We perform a simulation study. In all the test
cases we observe that designs with high values of $D$-efficiency
have also high values of the new objective functions.}

\abstract{In this work we focus on saturated $D$-optimal designs.
Using  recent results, we identify $D$-optimal
designs with the solutions of an optimization problem with linear
constraints. We introduce new objective functions based on the geometric structure
of the design and we compare them  with the classical  $D$-efficiency criterion.
 We perform a simulation study. In all the test
cases we observe that designs with high values of $D$-efficiency
have also high values of the new objective functions.}

\section{Introduction}

The optimality of an experimental design depends on the statistical model that
is assumed and is assessed with respect to a statistical
criterion. Among the different criteria, in this chapter we focus
on $D$-optimality.

Widely used statistical systems like {\tt SAS} and {\tt R} have
procedures for finding an optimal design according to the user's
specifications. {\tt Proc Optex} of {\tt SAS/QC} \cite{man3}
searches for optimal experimental designs in the following way.
The user specifies an efficiency criterion, a set of candidate
design points, a model and the size of the design to be found, and
the procedure generates a subset of the candidate set so that the
terms in the model can be estimated as efficiently as possible.

There are several algorithms for searching for $D$-optimal
designs. They have a common structure. Indeed, they start from an
initial design, randomly generated or user specified, and move, in
a finite number of steps, to a better design. All of the search
algorithms are based on adding  points to the growing
design and deleting points from a design that is too big.  Main references to optimal
designs include \cite{atkinson2007optimum},
\cite{goos2011optimal}, \cite{pukelsheim2006optimal},
\cite{rasch2011optimal}, \cite{shah1989theory} and
\cite{wynn1970sequential}.

In this work, we perform a simulation study to analyze a different
approach for describing $D$-optimal designs in the case of
saturated fractions. Saturated fractions, or saturated designs, contain a number of points
that is equal to the number of estimable parameters of the model. It follows
that saturated designs are often used in place of standard
designs, such as orthogonal fractional factorial designs, when the
cost of each experimental run is high. We show how the geometric
structure of a fraction is in relation with its $D$-optimality,
using a recent result in \cite{jspi2013} that allows us to
identify saturated designs with the points with coordinates in
$\{0,1\}$ of a polytope, being the polytope described by a system
of linear inequalities. The linear programming problem is based on
a combinatorial object, namely the circuit basis of the model
matrix. Since the circuits yield a geometric characterization of
saturated fractions, we investigate here the connections between
the classical $D$-optimality criterion and the position of the
design points with respect to the circuits.

In this way the search for $D$-optimal designs can be stated as an
optimization problem where the constraints are a system of linear
inequalities.  Within the classical framework the
objective function to be maximized is the determinant of the
information matrix. In our simulations, we define new objective
functions, which take into account the geometric structure of the
design points with respect to the circuits of the relevant design
matrix. We study the behavior of such objective functions and we
compare them with the classical $D$-efficiency criterion.

The chapter is organized as follows. In Sect.~\ref{sec:1} we briefly
describe the results of \cite{jspi2013} and in particular how
saturated designs can be identified with $\{0,1\}$ points that
satisfy a system of linear inequalities. Then in Sect.~\ref{sec:2}
we present the results of a simulation study in which, using some
test cases, we experiment different objective functions and we
analyze their relationship with the $D$-optimal criterion.
Concluding remarks are made in Sect.~\ref{sec:conclusion}.

\section{Circuits and saturated designs}
\label{sec:1}

As described in \cite{jspi2013}, the key ingredient to
characterize the saturated fractions of a factorial design is its
circuit basis. We recall here only the basic notions about
circuits in order to introduce our theory. For a survey on
circuits and its connections with Statistics, the reader can refer
to \cite{ohsugi:12}.

Given a model matrix $X$ of a full factorial design ${\mathcal
D}$, an integer vector $f$ is in the kernel of $X^t$ if and only
if $X^t f=0$. We denote by $A$ the transpose of $X$. Moreover,
we denote by ${\mathrm{supp}}(f)$ the support of the integer
vector $f$, i.e., the set of indices $j$ such that $f_j \ne 0$.
Finally, the indicator vector of $f$ is the binary vector $(f_j
\ne 0)$, where $( \cdot )$ is the indicator function. An integer
vector $f$ is a circuit of $A$ if and only if:
\begin{enumerate}
\item $f \in \ker(A)$;
\item there is no other integer vector $g \in {\ker(A)}$ such that
${\rm supp}(g) \subset {\rm supp}(f)$ and ${\rm supp}(g) \ne {\rm
supp}(f)$.
\end{enumerate}

The set of all circuits of $A$ is denoted by ${\mathcal C}_A$, and
is named as the circuit basis of $A$. It is known that ${\mathcal
C}_A$ is always finite. The set ${\mathcal C}_A$ can be computed
through specific software. In our examples, we have used {\tt
4ti2} \cite{4ti2}.

Given a model matrix $X$ on a full factorial design ${\mathcal D}$
with $K$ design points and $p$ degrees of freedom, we recall that a
fraction ${\mathcal F} \subset {\mathcal D}$ with $p$ design
points is saturated if $\det(X_{\mathcal F}) \ne 0$, where
$X_{\mathcal F}$ is the restriction of $X$ to the design points in
${\mathcal F}$. With a slight abuse of notation, ${\mathcal F}$
denotes both a fraction and its support. Under these assumptions,
the relations between saturated fractions and the circuit basis
${\mathcal C}_A = \{f_1, \ldots, f_L\}$ associated to $A$ is
illustrated in the theorem below, proved in \cite{jspi2013}.

\begin{theorem} \label{mainthm}
${\mathcal F}$ is a saturated fraction if and only if it does not
contain any of the supports $\{\mathrm{supp}(f_1), \ldots,
\mathrm{supp}(f_L)\}$ of the circuits of $A=X^t$.
\end{theorem}

\section{Simulation study}
\label{sec:2}

The theory described in Sect.~\ref{sec:1} allows us to identify
saturated designs with the feasible solutions of an integer linear
programming problem. Let $C_A=(c_{ij}, i=1,\ldots,L,
j=1,\ldots,K)$ be the matrix, whose rows contain the values of the
indicator functions of the circuits $f_1, \ldots, f_L$,
$c_{ij}=(f_{ij} \ne 0) , i=1,\ldots,L, j=1,\ldots,K$ and
$Y=(y_1,\ldots,y_K)$ be the $K$-dimensional column vector that
contains the unknown values of the indicator function of the
points of ${\mathcal F}$. In our problem the vector $Y$ must satisfy the
following conditions:
\begin{enumerate}
\item the number of points in the fractions must be equal to $p$;
\item the support of the fraction must not contain any of the
supports of the circuits.
\end{enumerate}
In formulae, this fact translates into the following constraints:
\begin{equation} \label{eq:xspp1}
{1}_K^t Y =p ,
\end{equation}
\begin{equation} \label{eq:xspp2}
C_A Y < b,
\end{equation}
\begin{equation} \label{eq:xspp3}
y_i \in \{0,1\}
\end{equation}
where $b=(b_1, \ldots, b_L)$ is the column vector defined by
$b_i=\#{\rm supp}(f_i), i=1,\ldots , L$, and ${1}_K$ is the
column vector of length $K$ and whose entries are all equal to $1$.

Since $D_Y= \det(V(Y)) = \det( X^t_{\mathcal F} X_{\mathcal F})$
is an objective function, it follows that a $D$-optimal design is
the solution of the optimization problem
\begin{eqnarray*}
\text{ maximize } \det(V(Y))\\
\text{subject to } (\ref{eq:xspp1}), (\ref{eq:xspp2}) \text{ and }  (\ref{eq:xspp3}).
\end{eqnarray*}

In general the objective function to be maximized $\det(V(Y))$ has several local optima and the problem
of finding the global optimum is part of current research, \cite{fontana_dopt}.
Instead of trying to solve this optimization problem in this work we prefer to study different
objective functions that are simpler than the original one
but that could generate the same optimal solutions. By analogy of
Theorem \ref{mainthm}, our new objective functions are defined
using the circuits of the model matrix.

For any $Y$, we define the vector $b_Y=C_A Y$. This vector $b_Y$
contains the number of points that are in the intersection between
the fraction $\mathcal{F}$ identified by $Y$ and the support of
each circuit $f_i \in {\mathcal C}_A, i=1,\ldots,L$. From
(\ref{eq:xspp2}) we know that each of these intersections must be
strictly contained in the support of each circuit. For each
circuit $f_i, i=1,\ldots,L$ it seems natural to minimize the
cardinality  $(b_Y)_i$ of the intersection between its support ${\rm
supp}(f_i)$ and Y with respect to the size of its
support, $b_i$. Therefore, we considered the following two
objective functions:
\begin{itemize}
\item $g_1(Y) = \sum_{i=1}^L (b-b_Y)_i$;

\item $g_2(Y) = \sum_{i=1}^L (b-b_Y)_i^2$.
\end{itemize}

From the examples analyzed in Sect.~\ref{sec:2.1}, we observe that
the $D$-optimality is reached with fractions that contain part of
the largest supports of the circuits, although this fact seems to disagree with Thm.~\ref{mainthm}. In fact, Thm.~\ref{mainthm} states that fractions containing the support of a circuit are not saturated, and therefore one would expect that optimal fractions will have intersections as small as possible with the supports of the circuits. On the other hand, our experiments show that optimality is reached with fractions having intersections as large as possible with such supports. For this
reason we consider also the following objective function:
\begin{itemize}
\item $g_3(Y) = \max(b_Y)$.
\end{itemize}

As a measure of $D$-optimality we use the $D$-efficiency,
\cite{man3}. The  $D$-efficiency of a fraction ${\mathcal F}$ with
indicator vector $Y$ is defined as
\[
E_{Y} = \left( \frac{1}{{\#{\mathcal F}}}
D_{Y}^{\frac{1}{{\#{\mathcal F}}}} \right) \times 100
\]
where $\#{\mathcal F}$ is the number of points of ${\mathcal F}$
that is equal to $p$ in our case, since we consider only saturated
designs.

\subsection{First case. $2^4$ with main effects and $2$-way interactions}
\label{sec:2.1}

Let us consider the $2^4$ design and the model with main factors
and 2-way interactions. The design matrix $X$ of the full design
has $16$ rows and $11$ columns, the number of estimable
parameters. As the matrix $X$ has rank $11$, we search for
fractions with $11$ points. A direct computation shows that there
are $\binom{16}{11}=4,368$ fractions with $11$ points: among them
$3,008$ are saturated, and the remaining $1,360$ are not.  Notice that equivalences up to permutations of factor or levels are not considered here.

The circuits are $140$ and the  cardinalities of their supports are $8$ in $20$ cases, $10$ in $40$ cases, $12$ in $80$ cases.
 For more details refer to \cite{jspi2013}.
This example is small enough for a complete enumeration of all
saturated fractions. Moreover, the structure of that fractions
reduces to few cases, due to the symmetry of the problem.

For each saturated fraction $\mathcal{F}$ with indicator vector
$Y$ we compute the vector $b_Y$, whose components are the size of
the intersection between the fraction and the support of all the
circuits, $\mathcal{F}\cap {\rm supp}(f_i), i=1,\ldots, 140$, and
we consider $b-b_Y$. Recall that $b$ is the vector of the
cardinalities of the circuits. The frequency table of $b-b_Y$
describes how many points need to be added to a fraction in order
to complete each circuit. All the frequency tables are displayed
in the left side of Table \ref{piccola}, while on the right side
we report the corresponding values of $D$-efficiency.

\begin{table}  \caption{Frequency tables of $b-b_Y$ for the $2^4$ design with main effects
and $2$-way interactions.}
\label{piccola}       
%
%
\begin{tabular}{p{1cm}p{1cm}p{1cm}p{1cm}p{1cm}p{1cm}p{1cm}p{1cm}p{1cm}}
\hline\noalign{\smallskip}
 \multicolumn{5}{c}{table$(b-b_Y)$} & & \multicolumn{3}{c}{$E_Y$}  \\
\noalign{\smallskip}\svhline\noalign{\smallskip}
$1$ & $2$  & $3$ &  $4$ & $5$ & & $68.29$ &  $77.46$ & $83.38$  \\
\noalign{\smallskip}\svhline\noalign{\smallskip}
  5 & 15 & 50 & 60 & 10 & & 192  &    0 &    0 \\
  5 & 18 & 48 & 55 & 14 & & 1,040 &     0 &   0 \\
  5 & 21 & 46 & 50 & 18 & & 960  &  0  &  0 \\
  5 & 24 & 44 & 45 & 22 & & 480  &  0 &   0 \\
  5 & 27 & 42 & 40 & 26 & & 0  & 320 &   0 \\
  5 & 30 & 40 & 35 & 30 & & 0  &  0 &  16  \\
\noalign{\smallskip}\hline\noalign{\smallskip}
  & & & & & {Total} &  2,672 & 320 & 16 \\
\noalign{\smallskip}\hline\noalign{\smallskip}
\end{tabular}
\end{table}

For instance, consider one of the $192$ fractions in the first
row. Among the $140$ circuits, $5$ of them are completed by adding
$1$ point to the fraction, $15$ of them by adding $2$ points, and
so on. We observe that there is a perfect dependence between the $D$-efficiency
and the frequency table of $b-b_Y$.

However, analyzing the objective functions $g_1(Y)$, $g_2(Y)$ and
$g_3(Y)$, we argue that the previous finding has no trivial
explanation. The values of all our objective functions are
displayed in Table \ref{tab:sat_des_2_4}.

\begin{table}  \caption{Classification
of all saturated fractions for the $2^4$ design with main effects
and $2$-way interactions.}
\label{tab:sat_des_2_4}       
%
%
\begin{tabular}{p{2cm}p{2cm}p{2cm}p{2cm}p{2cm}}
\hline\noalign{\smallskip}
$g_1(Y)$ & $g_2(Y)$  & $g_3(Y)$ &  $E_Y$ & $n$  \\
\noalign{\smallskip}\svhline\noalign{\smallskip}
475 & 1,725 & 9 & 68.29 & 192 \\
475 & 1,739 & 10 & 68.29 & 960 \\
475 & 1,753 & 10 & 68.29 & 960 \\
475 & 1,739 & 11 & 68.29 & 80 \\
475 & 1,767 & 11 & 68.29 & 480 \\
475 & 1,781 & 11 & 77.46 & 320 \\
475 & 1,795 & 11 & 83.38 & 16 \\
\noalign{\smallskip}\hline\noalign{\smallskip}
& & & Total & 3,008 \\
\noalign{\smallskip}\hline\noalign{\smallskip}
\end{tabular}
\end{table}

From Table \ref{tab:sat_des_2_4} we observe that both $g_2(Y)$ and
$g_3(Y)$ are increasing as $D$-efficiency increases. Notice also
that $g_1(Y)$ is constant over all the saturated fractions. This
is a general fact for all no-$m$-way interaction models.

\begin{proposition} \label{smallprop}
For a no-$m$-way interaction model, $g_1(Y)$ is constant over all
saturated fractions.
\end{proposition}
\begin{proof}
We recall that $C_A=(c_{ij}, i=1,\ldots,L, j=1,\ldots,K)$ is the
$L \times K$ matrix, whose rows contain the values of the
indicator functions of the supports of the circuits $f_1, \ldots,
f_L$, $c_{ij}=(f_{ij} \ne 0) , i=1,\ldots,L, j=1,\ldots,K$. We
have
\[
g_1(Y)= \sum_{i=1}^L (b-b_Y)_i=  \sum_{i=1}^L (b)_i - \sum_{i=1}^L
(b_Y)_i \, .
\]
The first addendum does not depend on $Y$, and for the second one
we get
\[
\sum_{i=1}^L (b_Y)_i= \sum_{i=1}^L \sum_{j=1}^K  c_{ij} Y_j =
\sum_{j=1}^K Y_j \sum_{i=1}^L c_{ij} \, .
\]
Now observe that a no-$m$-way interaction model does not change
when permuting the factors or the levels of the factors.
Therefore, by a symmetry argument, each design point must belong
to the same number $q$ of circuits, and thus $\sum_{i=1}^L
c_{ij}=q$. It follows that
\[
\sum_{i=1}^L (b_Y)_i = q \sum_{j=1}^K Y_j = pq \, .
\]
\qed
\end{proof}

In view of Prop.~\ref{smallprop}, in the remaining examples we
will consider only the functions $g_2$ and $g_3$.

\subsection{Second case. $3 \times 3 \times 4$ with main effects and $2$-way interactions}

Let us consider the $3 \times 3 \times 4$ design and the model
with main factors and 2-way interactions. The model has $p=24$
degrees of freedom. The number of circuits is $17,994$. In this
case the number of possible subsets of the full design is
$\binom{36}{24}=1,251,677,700$. It would be computationally
unfeasible to analyze all the fractions. We use the methodology
described in \cite{fontana_dopt} to obtain a sample of saturated
$D$-optimal designs. It is worth noting that this methodology
finds $D$-optimal designs and not simply saturated designs. This
is particularly useful in our case because allows us to study
fractions for which the $D$-efficiency is very high. The sample
contains $500$ designs, $380$ different.

The results are summarized in Table \ref{tab:sat_des_3_3_4}, where
the fractions with minimum $D$-efficiency $E_Y$ have been
collapsed in a unique row in order to save space. We observe that
for $138$ different designs the maximum value of $D$-efficiency,
$E_Y=24.41$ is obtained for both $g_2(Y)$ and $g_3(Y)$ at their
maximum values $g_2(Y)=970,896$ and $g_3(Y)=24$.

\begin{table}
\caption{Classification of $380$ random saturated fractions for
the $3 \times 3 \times 4$ design with main effects and $2$-way
interactions.}
\label{tab:sat_des_3_3_4}       
%
%
\begin{tabular}{p{2cm}p{2cm}p{2cm}p{2cm}}
\hline\noalign{\smallskip}
$g_2(Y)$  & $g_3(Y)$ &  $E_Y$ & $n$  \\
\noalign{\smallskip}\svhline\noalign{\smallskip}
$\leq$963,008 & $\leq$21 & 22.27 & 37 \\
962,816 & 21 & 23.6 & 7 \\
962,816 & 22 & 23.6 & 12 \\
963,700 & 22 & 23.6 & 34 \\
965,308 & 22 & 23.6 & 46 \\
966,760 & 22 & 23.6 & 9 \\
967,676 & 22 & 23.6 & 6 \\
970,860 & 24 & 23.6 & 91 \\
970,896 & 24 & 24.41 & 138 \\
\noalign{\smallskip}\hline\noalign{\smallskip}
& & Total & 380 \\
\noalign{\smallskip}\hline\noalign{\smallskip}
\end{tabular}
\end{table}

\subsection{Third case. $2^5$ with main effects}
Let us consider the $2^5$ design and the model with main effects
only. The model has $p=6$ degrees of freedom. The number of
circuits is $353,616$. As in the previous case we use the
methodology described in \cite{fontana_dopt} to get a sample of
$500$ designs, $414$ different.

The results are summarized in Table \ref{tab:sat_des_2_5}. We
observe that for $194$ different designs, the maximum value of
$D$-efficiency, $E_Y=90.48$ is obtained for both $g_2(Y)$ and
$g_3(Y)$ at their maximum values $g_2(Y)=11,375,490$ and $g_3(Y)=6$.

%
\begin{table}
\caption{Classification of $414$ random saturated fractions for
the $2^5$ design with main effects.}
\label{tab:sat_des_2_5}       
%
%
\begin{tabular}{p{2cm}p{2cm}p{2cm}p{2cm}}
\hline\noalign{\smallskip}
$g_2(Y)$  & $g_3(Y)$ &  $E_Y$ & $n$  \\
\noalign{\smallskip}\svhline\noalign{\smallskip}
11,360,866 & 6 & 76.31 & 31 \\
11,342,586 & 6 & 83.99 & 9 \\
11,371,834 & 6 & 83.99 & 126 \\
11,375,490 & 5 & 83.99 & 54 \\
11,375,490 & 6 & 90.48 & 194 \\
\noalign{\smallskip}\hline\noalign{\smallskip}
& & Total & 414 \\
\noalign{\smallskip}\hline\noalign{\smallskip}
\end{tabular}
\end{table}


\section{Concluding remarks}
\label{sec:conclusion}

The examples discussed in the previous section show that the
$D$-efficiency of the saturated fractions and the new objective
functions based on combinatorial objects are strongly dependent.
The three examples suggest to investigate such connection in a
more general framework, in order to characterize saturated
$D$-optimal fractions in terms of their geometric structure.
Notice that our presentation is limited to saturated fractions,
but it would be interesting to extend the analysis to other kinds
of fractions. Moreover, we need to investigate the connections
between the new objective functions and other criteria than
$D$-efficiency.

Since the number of circuits dramatically increases with the
dimensions of the factorial design, both theoretical tools and
simulation will be essential for the study of large designs.

\bibliographystyle{spmpsci}
\bibliography{referenc_ROBERTOFONTANA}

\end{document}